\documentclass[a4paper,11pt]{article}
\usepackage[utf8]{inputenc}
\usepackage{amsmath,amsfonts,amssymb,amsthm}
\usepackage{tensor,graphicx}

\usepackage{natbib}
\usepackage{hyperref}
\hypersetup{
    colorlinks=true,
    citecolor=blue,
    linkcolor=blue,
    filecolor=magenta,      
    urlcolor=cyan,
}
\usepackage{xcolor}
\theoremstyle{plain}
\newtheorem{theorem}{Theorem}[section]

\theoremstyle{definition}

\theoremstyle{remark}
\newtheorem*{remark}{Remark}

\title{A new generalized newsvendor model with random demand  }
\author{Soham Ghosh, Mamta Sahare and Sujay Mukhoti\footnote{sujaym@iimidr.ac.in}\\ Operations Management and Quantitative Techniques Area\\ Indian Institute of Management, Indore \\ Rau-Pithampur Road, Rau, Indore, Madhya Pradesh, India- 453556}
\date{}

\begin{document}
\maketitle

\vspace{12pt}
\begin{abstract}	
	  Newsvendor problem is an extensively researched topic in inventory management. In this class of inventory problems, shortage and excess costs are considered to be proportional to the quantity lost. But, for critical goods or commodities, inventory decision is a typical example where, excess or shortage may lead to greater losses than merely the total cost. Such a problem has not been discussed much in the literature. Moreover, majority of the existing literature assumes the demand distribution to be completely known. In this paper, we propose a generalization of the newsvendor problem for critical goods or commodities with higher shortage or excess losses but of same degree. We also assume that, the parameters of the demand distribution are unknown. We also discuss different estimators of the optimal order quantity based on a random sample of demand. In particular, we provide different estimators based on (i) full sample and (ii) broken sample data (i.e with single order statistic). We also report comparison of the estimators using simulated bias and mean square error (MSE).
\end{abstract}

\section{Introduction}
\label{sec:1}
Newsvendor problem is one of the most extensively discussed problems in the inventory management literature due to its applicability in different fields \citep{silver1998}. This inventory problem relies on offsetting the shortage and leftover cost in order to obtain optimal order quantity. In the standard newsvendor problem, each of the shortage and leftover costs are assumed to be proportional to the loss amount, i.e linear in the order quantity and demand. Also, the demand is assumed to be a random variable with known probability distributions. However, there may be situations where the loss is more severe than the usual linear one. Further in real life situations, the demand distribution is often unknown. In this paper, we generalize the classical newsvendor problem with non-linear losses and study the estimation of optimal order quantity subject to random demand with unknown parameters.

In classical newsvendor problem, we consider a newsvendor selling a perishable commodity procured from a single supplier. Further let the demand be a random variable. Newsvendor takes one time decision on how much quantity she should order from supplier. The newsvendor faces leftover cost if the demand is lower than the inventory, as a penalty for ordering too much. Similarly, shortage cost is faced if the demand is higher than the inventory, as a penalty for ordering too less. The shortage (leftover) cost is computed in terms of currency units as the product of per unit shortage cost and the loss amount, i.e difference between demand and inventory. Thus the loss is linear in difference between demand and inventory. To determine the optimal order quantity, the newsvendor has to minimize the total cost or maximize the total profit. This classical version of the newsvendor problem has been generalized, application wise, in many directions since its inception. \citet{veinott1965} generalized the newsvendor model for multiple time periods. Extension of newsvendor problem for multiple products is being studied extensively in the literature \citep[e.g., see][and the references therein]{chernonog2018}. We refer to a recent review by \citep{qin2011} for extensions of newsvendor problem that adds dimensions like marketing efforts, buyer's risk appetite etc.

In many real life situations, the notion of loss may be higher than the quantity of lost units. For example, chemotherapy drugs are critical for administering it to a patient on the scheduled days. Shortage of the drug of that day would result in breaking the cycle of treatment. Hence the loss is not merely the quantity but more than that. Similarly, in the case of excess inventory, not only the excess amount of the drug but its disposal method also contributes to the notion of loss. This is due to the chance of vast environmental and microbial hazards that may be created through improper disposal, like creation anti-biotic resistant bacteria or super-bugs. Thus the losses due to excess and shortage could be thought of as non-linear. ~Newsvendor problem with nonlinear losses, however, remains not much addressed until recently. 

\citet{chandra2005} have considered optimization of different risk alternatives to expected cost function, like cost volatility and Value-at-Risk, which results in non-linear objective functions. \citet{parlar1992} considered the periodic review inventory problem and derived the solution for a quadratic loss function for both shortage and leftover. \citet{gerchak1997} proposed a newsvendor model using power type loss function for asset allocation. Their work assumes asymmetric losses with linear leftover but power type shortage loss with details only up to quadratic loss. In this paper we consider a generalization of the newsvendor problem in the line of \citet{parlar1992} and \citet{gerchak1997}, i.e. power type losses. In particular, we consider the shortage and excess losses as general power function of same degree. Determination of optimal order quantity from such a generalized newsvendor problem still remains unexplored to the best of our knowledge. Since the power of excess and shortage loss are same, we refer to this problem as symmetric generalized (SyGen) newsvendor problem.

An interesting observation that can be made from the newsvendor literature is that majority of related work considers a completely specified demand distribution, whereas in reality, it is seldom known. In such cases, the optimal order quantity needs to be estimated. \citet{dvoretzky1952} first addressed the estimation problem in classical newsvendor setup using Wald's decision theoretic approach. \citet{scarf1959, hayes1969,fukuda1960estimation} considered estimation problems in inventory control using maximum likelihood estimation (MLE) under different parametric demand distributions. \citet{conrad1976} estimated the demand and hence the optimal order quantity for Poisson distribution. \citet{nahmias1994demand} estimated the optimal order quantity for Normal demand with unknown parameters and \cite{agrawal1996estimating} estimated the order quantity for negative binomial demand.  \citet{sok1980} in his master's thesis, presented estimators of the optimal order quantity based on order statistics for parametric distributions including uniform and exponential. \citet{rossi2014confidence} has given bounds on the optimal order quantity using confidence interval for parametric demand distributions.  

On the other hand, non-parametric estimation of optimal order quantity in a classical newsvendor problem is comparatively recent. For example, \citet{bookbinder1998} considered bootstrap estimator of the optimal order quantity and \citet{pal1996} discussed construction of asymptotic confidence interval of the cost using bootstrapping. Another important non-parametric data-driven approach is the sampling average approximation (SAA) \citep{shapiro2001}. In this method, the expected cost is replaced by the sample average of the corresponding objective function and then optimized. \citet{levi2015data} provides bounds of the relative bias of estimated optimal cost using SAA based on full sample data. \citet{bertsimas2005data} ranks the objective functions evaluated at each demand sample data and shows that the trimmed mean of the ordered objective functions leads to a convex problem ensuring robust and tractable solution.

In this paper we consider estimation of the optimal order quantity for parametric demand distributions under SyGen newsvendor set-up.  In particular, we consider two specific demand distributions, viz. uniform and exponential. We present here the optimal order quantity and its estimators using (1) full data in a random sample and (2) order statistic from a full sample. We investigate the conditions for existence of the estimator of the optimal order quantity. Further, we do a simulation study to gauge the performance of different estimators in terms of bias and mean square error (MSE) for different shortage to excess cost ratio, degree of loss-importance and sample size.

Rest of this paper is organized as follows. Section 2. describes determination of optimal order quantities for uniform and exponential demand distribution using full sample and order statistics, along with the existence condition for the same, wherever required. In section 3, we investigate estimation of the optimal order quantity using full sample and order statistics. Results of simulation study is presented in section 4. Section 5 is the concluding section with discussion on the future problems.

\section{SyGen: Symmetric generalized newsvendor problem}
\label{sec:2}

In a single period classical newvendor problem, the vendor has to order the inventory before observing the demand so that the excess and shortage costs are balanced out. Let us denote the demand by a positive random variable $X$ with cumulative distribution function (CDF) $F_\theta(x), \; \theta\in \Theta$ and probability density function $f_\theta(x)$. We further assume that the first moment of $X$ exists finitely. Suppose $Q \in \mathbb{R}^+$ is the inventory level at the beginning of period. Then, shortage loss is defined as $(X-Q)^+$ and the excess loss is defined as $(Q-X)^+$, where $d^+=max(d,0)$. Let $C_s$ and $C_e$ denote the per unit shortage and excess costs (constant) respectively and the corresponding costs are defined as $C_s(X-Q)^+$ and $C_e(Q-X)^+$. Thus the total cost can be written as a piece-wise linear function in the following manner:
\begin{eqnarray}
\chi=\left\{ \begin{array}{cc}
C_s(X-Q) & \mbox{ if }X > Q \\
C_e(Q-X) & \mbox{ if } X\leq Q
\end{array}	\right.
\end{eqnarray}
The optimal order quantity ($Q^*$) is obtained by minimizing the expected total cost, $E[\chi]$. The analytical solution in this problem is given by $Q^*=F_\theta^{-1}\left( \frac{C_s}{C_e+C_s}\right)$, that is the $\gamma^{th}$ quantile of the demand distribution ($\gamma=\frac{C_s}{C_s+C_e}$). 

We propose the following extension the classical newsvendor problem replacing the piece-wise linear loss functions by piece-wise power losses of same degree. Thus, the generalized cost function is given by
\begin{eqnarray}
\chi_m=\left\{ \begin{array}{cc}
C_s(X-Q)^{m} & \mbox{ if }X > Q \\
C_e(Q-X)^{m} & \mbox{ if } X\leq Q
\end{array}	\right.
\end{eqnarray}
The powers of losses on both sides of $Q$ on the support of $X$ being same we would refer to this problem as symmetric generalized (SyGen) newsvendor problem.

The expected total cost in SyGen newsvendort problem is given by,
\[ E[\chi_m] = {\int_0^{Q} C_e(Q-x)^m f_{\theta}(x)dx} + {\int_{Q}^\infty C_s(x-Q)^m f_{\theta}(x)dx}\]
In the subsequent sections, we show that the expected cost admits minimum for uniform and exponential distributions.

In order to do so, the first order condition (FOC) is given by
\begin{eqnarray}
& &\frac{\partial E[\chi_m]}{\partial Q}  =  \int_0^Q mC_e(Q-x)^{m-1} f_{\theta}(x)dx - \int_{Q}^\infty mC_s(x-Q)^{m-1} f_{\theta}(x)dx=0 \nonumber \\
& \Rightarrow &  C_e{\int_0^Q (Q-x)^{m-1} f_{\theta}(x)dx}  =   C_s {\int_{Q}^\infty (x-Q)^{m-1} f_{\theta}(x)dx}
\label{FOCGen}
\end{eqnarray}

It is difficult to provide further insight without having more information on $f_\theta(x)$. In the following section, we assume two choices for the demand distribution - (i) \emph{Uniform} and (ii) \emph{Exponential}.

\subsection{Optimal Order Quantity for  SyGen Newsvendor with Uniform Demand}

In this section, we consider the problem of determining optimal order quantity in SyGen newsvendor set up where the demand is assumed to be a $Uniform$ random variable over the support $(0,b)$ . The pdf of demand distribution is given by,
\[f_{\theta}(x)= \left\{\begin{array}{cc}
\frac{1}{b} & \; if\; 0<x<b \\
0;& \; otherwise
\end{array}\right. \]
The minimum demand here is assumed to be zero without loss of generality, because any non-zero lower limit of the support of demand could be considered as a pre-order and hence can be pre-booked. Using Leibnitz rule and routine algebra, it can be shown that the optimal order quantity is given by ${{Q^*=\frac{b}{1+\alpha_m}}}$, where $\alpha_m=\left(\frac{C_e}{C_s}\right)^\frac{1}{m}$. Corresponding optimal cost is  ${{C_s\times\frac{b^m}{m+1}\times \frac{C_e}{\left(C_e^{1/m}+C_s^{1/m}\right)^m}}}$. Notice that $\left(C_e^{1/m}+C_s^{1/m}\right)^m=C_s \left(1+\alpha_1^{1/m}\right)^m \geq C_s$, and hence we get an upper bound of the optimal cost as $C_e\frac{b^m}{m+1}$. 
Notice, large $\alpha_m$ would imply small order quantity. In other words, if the cost of excess inventory is much larger than the shortage cost, then the newsvendor would order less quantity to avoid high penalty. Similarly, small $\alpha_m$ would result in ordering closer to the maximum possible demand ($b$) to avoid high shortage penalty. However, if the degree of loss ($m$) is very high, then the optimum choice would be to order half the maximum possible demand, \emph{i.e.} $Q^*\rightarrow\frac{b}{2}, \; as \; m\rightarrow \infty$. 


\subsection{Optimal Order Quantity for  SyGen Newsvendor with Exponential Demand}
Next we consider the demand to be exponentially distributed with mean ${\lambda}$. The pdf of exponential distribution is given by,
\[f_{\lambda}(x)={\frac{1}{\lambda}}e^{-{\frac{x}{\lambda}}}; \; x>0,\; \lambda>0 \]
In this case the expected cost function becomes 
\begin{equation}
E[\chi_m] = {\int_0^{Q} C_e(Q-x)^m f_{\theta}(x)dx} + {\int_{Q}^\infty C_s(x-Q)^m f_{\theta}(x)dx} \label{ExpCost}
\end{equation}
In the following theorem we derive the first order condition for optimal inventory in SyGen newsvendor problem with exponential demand.
\begin{theorem}
	Let the demand in a SyGen newsvendor problem be an exponential random variable $X$ with mean $\lambda>0$. Then the first order condition for minimizing the expected cost is given by
	\begin{equation}
	\psi(Q/ \lambda)  =  e^{-{\frac{Q}{\lambda}}} \left[{\frac{C_s}{C_e}}-(-1)^m\right]
	\label{FOCExp}
	\end{equation}
	where $\displaystyle \psi\left( \frac{Q}{\lambda}\right)==\sum_{j=0}^{m-1}(-1)^j \left(\frac{Q}{\lambda}\right)^{m-j-1}\frac{1}{(m-j-1)!}$.
\end{theorem}
\begin{proof}
	Let us define, $I_m={\int_{0}^{Q} (Q-x)^{m-1} {\frac{1}{\lambda}}e^{-{\frac{x}{\lambda}}}dx}$ and $J_m={\int_{Q}^{\infty} (x-Q)^{m-1} {\frac{1}{\lambda}}e^{-{\frac{x}{\lambda}}}dx}$. Hence, the FOC in eq.(\ref{FOCGen}) becomes 
	\begin{equation}
	C_eI_m=C_sJ_m \label{FOCGenRed}
	\end{equation}
	Assuming $Q-x=u$, we get
	\begin{eqnarray*}
		I_m&=&{\frac{e^{-{\frac{Q}{\lambda}}}}{\lambda}} {\int_{0}^{Q} u^{m-1} e^{\frac{u}{\lambda}}du} \\
		& = & {\frac{e^{-{\frac{Q}{\lambda}}}}{\lambda}}\times  \left[{\lambda}e^{\frac{u}{\lambda}}u^{m-1}\right]^Q_0-{\lambda}(m-1)I_{m-1}, \hspace{4pt}\mbox{ (integrating by parts)}\\ 
		& = & Q^{m-1}-{\lambda}(m-1)Q^{m-2}+\lambda^2(m-1)(m-2)I_{m-2} \\
		& & \ldots ~ \ldots ~ \ldots \\
		&=& {\sum_{j=0}^{m-1}}Q^{m-1-j}(-{\lambda})^j {\frac{(m-1)!}{(m-j-1)!}}+e^{-{\frac{Q}{\lambda}}}{\lambda}^{m-1}(-1)^m(m-1)!
	\end{eqnarray*}
	Now letting, $x-Q=v$ in $J_m$, we get
	\begin{equation}
	J_m  =  {\frac{e^{-{\frac{Q}{\lambda}}}}{\lambda}} {\int_{0}^{\infty} v^{m-1} e^{-{\frac{v}{\lambda}}}dv}=e^{-{\frac{Q}{\lambda}}}{\Gamma}(m){\lambda}^{m-1} \nonumber
	\end{equation}
	Thus, from eq.(\ref{FOCGenRed}), we get 
	\begin{eqnarray}
	& & C_e I_m  =  C_s J_m \nonumber \\
	&{\Rightarrow} &  {\sum_{j=0}^{m-1}}(-1)^j{\left(\frac{Q}{\lambda}\right)^{m-j-1}}{\frac{1}{{(m-j-1)!}}} =e^{-{\frac{Q}{\lambda}}} \left[{\frac{C_s}{C_e}}-(-1)^m\right]= \gamma_m e^{-{\frac{Q}{\lambda}}}\nonumber
	\end{eqnarray}
\end{proof}
As a consequence of the above FOC condition for exponential demand, we need to inspect the existence of non-negative zeroes of eq.(\ref{FOCExp}). We first provide lower and upper bound of $\psi(Q/ \lambda)$ in the following theorem. 
\begin{theorem}
	For $ \gamma_m > 0 \; and \;  m\geq 4$, the following inequality holds:
	\begin{equation*}
	-(1+\gamma_m)\left(\frac{u^2}{2}-u+1\right)+ S_{m-4} < \psi(u)  < (1+\gamma_m)(u-1) +S_{m-3} 
	\end{equation*}
	where, $\psi(u)=S_{m-1} - e^{-u}\gamma_m$ and $\displaystyle S_{m-k}=\sum_{j=0}^{m-k}(-1)^j \frac{u^{m-j-1}}{(m-j-1)!}$.
\end{theorem}
\begin{proof}
	Letting $\frac{Q}{\lambda}=u$, we obtain from eq. (\ref{FOCExp}),
	\begin{equation*} \psi(u)=\sum_{j=0}^{m-1}(-1)^j \frac{u^{m-j-1}}{(m-j-1)!} - e^{-u}\gamma_m = 0, \; where, \; \gamma_m= \left[{\frac{C_s}{C_e}}-(-1)^m\right] 
	\end{equation*}
	Since, $e^{-u}>1-u,$, for $m\geq 3$, we get 
	\begin{eqnarray}
	\psi(u)& < & \sum_{j=0}^{m-1}(-1)^j \frac{u^{m-j-1}}{(m-j-1)!} -(1-u)\gamma_m \nonumber \\
	& = & ((-1)^{m-1}-\gamma_m)+((-1)^{m-2}+\gamma_m)u+\sum_{j=0}^{m-3}(-1)^j \frac{u^{m-j-1}}{(m-j-1)!}  \nonumber  \\
	& =  & ((-1)^{m-1}-\gamma_m)+((-1)^{m-2}+\gamma_m)u+\sum_{j=0}^{m-3}(-1)^j \frac{u^{m-j-1}}{(m-j-1)!}  \nonumber  \\
	& =  & ((-1)^{m-2}+\gamma_m)(u-1)+\sum_{j=0}^{m-3}(-1)^j \frac{u^{m-j-1}}{(m-j-1)!}  \nonumber  \\
	&\leq&  (1+\gamma_m)(u-1) +S_{m-3},\; [\mbox{ since, } (-1)^{m-2}\leq 1]
	\label{LboundExp}
	\end{eqnarray} 
	On the other hand,using $ e^{-u}<1-u+\frac{u^2}{2}$, it can be shown from eq. (\ref{FOCExp}) that for $m\geq 4$,
	\begin{eqnarray}
	\psi(u) & \geq & \sum_{j=0}^{m-1}(-1)^j \frac{u^{m-j-1}}{(m-j-1)!} -\left(1-u+\frac{u^2}{2}\right)\gamma_m \nonumber \\
	& = & ((-1)^{m-1}-\gamma_m)+((-1)^{m-2}+\gamma_m)u+((-1)^{m-3}-\gamma_m)\frac{u^2}{2} \nonumber \\
	& & + \sum_{j=0}^{m-4}(-1)^j \frac{u^{m-j-1}}{(m-j-1)!} \nonumber \\
	& = & ((-1)^{m-1}-\gamma_m)-((-1)^{m-1}-\gamma_m)u+((-1)^{m-1}-\gamma_m)\frac{u^2}{2} \nonumber \\
	& & + \sum_{j=0}^{m-4}(-1)^j \frac{u^{m-j-1}}{(m-j-1)!} \nonumber \\
	& = & ((-1)^{m-1}-\gamma_m)(1-u+\frac{u^2}{2}) \nonumber \\
	& & + \sum_{j=0}^{m-4}(-1)^j \frac{u^{m-j-1}}{(m-j-1)!} \nonumber \\
	& \geq & ((-1)-\gamma_m)(1-u+\frac{u^2}{2})  + S_{m-4} ,\; [\mbox{ since, } (-1)^{m-1}\geq -1]\nonumber \\
	& = & -(1+\gamma_m)(1-u+\frac{u^2}{2})  + S_{m-4} 
	\label{UboundExp}
	\end{eqnarray}
	
\end{proof}
The following few remarks could be made immediately from the above theorem. 
\begin{remark}
	{For $m=2k$, $\gamma_m<0$ implies $C_s<C_e$, in which case the direction of the above inequality will be altered. Further, for $m=2k+1$, $\gamma_m $ is always positive as $C_s, C_e>0$.} 
\end{remark}

\begin{remark}
	If $m=2k+1$, $\gamma_m>0$. 
	%
	In that case, we get from the lower boundary function in eq.(\ref{LboundExp}),
	\begin{equation*}
	g_L(u)={\frac{u^{2k}}{(2k)!}}-{\frac{u^{2k-1}}{(2k-1)!}}+...+{\frac{u^4}{4!}}-{\frac{u^{3}}{3!}}-(1+{\gamma}_m){\frac{u^2}{2}}+u(1+{\gamma}_m)-(1+{\gamma}_m) 
	\end{equation*} 
	Therefore, by Descarte's rule of sign, the maximum number of positive real roots is $2k-1$ and hence there will exist at least one positive root. 
	
	Further, from the upper boundary function in eq.(\ref{UboundExp})
	\begin{equation*}
	g_U(u)={\frac{u^{2k}}{(2k)!}}-{\frac{u^{2k-1}}{(2k-1)!}}+...+{\frac{u^4}{4!}}-{\frac{u^{3}}{3!}}+{\frac{u^2}{2}}+u(1+{\gamma}_m)-(1+{\gamma}_m) 
	\end{equation*}
	Hence, by similar argument as in $g_L(u)$, at least one positive root of $g_U(u)$ will exist. 
	Thus, both the boundary functions will have a positive root leading to the existence of positive root of $\psi(u)$.
\end{remark}
\begin{remark}
	If $m \; (\geq 4)$ is even then, following the similar line of argument as in the previous remark, it can be shown that both $g_L(u)$ and $g_U(u)$ will have $2k-1$ sign changes and hence at least one positive root. 
\end{remark}
\begin{remark}
	
	In the particular case of $m=3$, it the boundary functions reduces to
	\begin{eqnarray*}
		g_L(U)& =& (2+{\frac{C_s}{C_e}})({\frac{u^2}{2}}-u+1) \\
		\mbox{\& } g_U(u)& = & (u-1)({\gamma}_{m}+1)+{\frac{u^2}{2}}
	\end{eqnarray*}
	Clearly, no real root of the lower boundary functions exist and the positive root of upper boundary function is given by $-(1+\gamma_{m})+\sqrt{\gamma_{m}^2+4\gamma_m+3}$. 
\end{remark}


Since eq. (\ref{FOCExp}) is a transcendental equation in Q, numerical methods would be required to find zeros, which in turn would provide the optimal order quantity. However, the solution would be dependent on the shortage and excess cost ratio. In what follows, we describe the nature of solutions in different scenarios for $\alpha_1=\frac{C_s}{C_e}$.

In case of equal shortage and excess costs ($C_s=C_e$), $\alpha_1=1$ and corresponding optimal order quantities ($Q^*$) are given in the table below for different $m$:
\begin{table}[ht]
	\centering
	\caption{Optimal Order Quantity for $C_s=C_e$.\vspace{12pt}}	
	
	\begin{tabular}{ |c|c| } 
		\hline m & $Q^*$ \\ \hline
		2 & ${\lambda}$ \\
		3 &  $1.3008{\lambda}$\\
		4 & $1.5961{\lambda}$ \\
		10 &  $3.33755{\lambda}$\\ 
		20 & $6.17753{\lambda}$ \\
		\hline
		
	\end{tabular}
	
\end{table}

On the other hand, if the shortage cost is much lower than the excess cost ($\alpha_1\rightarrow 0$), then the optimal order quantity ($Q^*$) is very small (tends to $0$). 
In case of $0<\alpha_1<1$, the optimum order quantity can be computed from the following equation:
\[\sum_{j=0}^{m-1}(-1)^j 
\left(\frac{Q}{\lambda}\right)^{m-j-1}{\frac{1}{{(m-j-1)!}}}  =  \gamma_m e^{-{\frac{Q}{\lambda}}} \]

The two figures in appendix B, fig.\ref{qstar} provide the optimal order quantity as a multiple of the average demand ($\lambda$) obtained for $\alpha_1 \in (0,1)$ and $m=2,3,4,5,10,20,30,40,50,100$.

%
Notice that the differences between optimal order quantities corresponding to different $\alpha_1$s reduce with increasing $m$. Here we may interpret $m$ as the degree of seriousness of the losses. Hence, the observation made above may also be restated in the following manner. As the degree of seriousness of the loss increases, the newsvendor becomes indifferent to both the losses. That is, beyond a risk level the newsvendor will not react much to the increase in any type of loss and order at a steady level, \emph{i.e } become risk neutral. 

In the following section we consider the parameters of the demand distributions discussed above, to be unknown and study the estimation of the optimal order quantity. 

\section{Estimation of the Optimal Order Quantity}
\label{sec:3}
In this section we consider the problem of estimating the optimal order quantity in SyGen newsvendor setup, based on a random sample of fixed size on demand, when the parameters of the two demand distributions discussed above are unknown. 

\subsection{Estimating Optimal Order Quantity for $U(0,b)$ Demand}
As in the previous section, we first consider the demand distribution to be $uniform(0,b)$, where $b$ is unknown. We elaborate on the estimation of optimal order quantity in SyGen news vendor problem with available iid demand observations $X_1, X_2 \ldots X_n$. Method of moment type estimator of the optimal order quantity could be constructed by plugging in the same for the unknown parameter $b$. Thus,
\begin{equation}
\hat{Q}_1=\frac{2\bar{x}}{1+\alpha_m}
\label{unifT1}
\end{equation}
which is an unbiased estimator as well. The variance of $\hat{Q}_1$ is given by $V(\hat{Q}_1)=\frac{b^2}{3n(1+\alpha_m)^2}$. In fact the uniformly minimum variance unbiased estimator (UMVUE) of the optimal order quantity can be obtained using the order statistics $X_{(1)}<X_{(2)} \ldots <X_{(n)}$. The UMVUE is given as follows:
\begin{equation}
\hat{Q}_2=\frac{(n+1)X_{(n)}}{n(1+\alpha_m)}
\label{unifT2}
\end{equation}
The variance of the UMVUE is given by $V(\hat{Q}_2)=\frac{b^2}{(1+\alpha_m)^2n(n+2)}$. Since $X_{(n)}$ is the maximum likelihood estimator (MLE) of $b$, the MLE of optimal order quantity can also be obtained as 
\begin{equation}
\hat{Q}_3=\frac{X_{(n)}}{1+\alpha_m}
\label{unifT3}
\end{equation}
Note that $\hat{Q}_3$ is biased with  $Bias(\hat{Q}_3)=-Q^*\frac{1}{n+1}$ and mean square error $(MSE)=Q^{*2}\frac{2}{(n+1)(n+2)}$. Comparing $\hat{Q}_1,\; \hat{Q}_2$ and $\hat{Q}_3$ in terms of their variances (MSE for $\hat{Q}_2$), it can be easily seen that the UMVUE provides the best estimator among the three. In particular, $V(\hat{Q}_2)<MSE(\hat{Q}_3)<V(\hat{Q}_1)\; \forall b>0$. 
\subsection{Estimating Optimal Order Quantity for exp($\lambda$) Demand}
Let us now consider that the demand distribution to be exponential with unknown parameter $\lambda>0$. Let $X_1,X_2,\ldots,X_n$ be a random sample on demand. Thus the problem becomes estimation of optimal order quantity in this SyGen setup. 

\vspace{12pt}
\subsubsection{Estimating equation based on full sample}
In order to provide a good estimator of the optimal order quantity based on the full sample, we replace the parametric functions of $\lambda$ involved in the FOC eq. (\ref{FOCExp}) by their suitable estimators. In particular, we focus on replacing (i) $\lambda$ by its MLE, (ii) $e^{-\frac{Q}{\lambda}}=\bar{F}(Q)$ by corresponding UMVUEs. Performances of the estimated optimal order quantities, $Q^*$, is measured by the corresponding bias and MSE. 

\vspace{12pt}
\textbf{FIRST ESTIMATING EQUATION} \\

Replacing ${\lambda}$ in eq.~(\ref{FOCExp}) by its MLE ${\bar{X}}$, the first estimating equation is obtained as follows: 
\begin{equation}
\sum_{j=0}^{m-1}(-1)^j\left( \frac{Q}{\bar{X}}\right)^{m-j-1}{\frac{1}{{(m-j-1)!}}} =e^{-{\frac{Q}{\bar{X}}}}\gamma_{m}
\label{eeq1}
\end{equation}
Let the solution of this estimating equation be denoted by $Q^*_1$. Though $Q^*_1$ is a plug-in estimator, being a function of MLE, it still would be expected to perform well in terms of bias and MSE.

\vspace{12pt}
\textbf{SECOND ESTIMATING EQUATION}\\

The UMVUE of $ e^{-\frac{Q}{\lambda}}={\bar{F}}(Q)$ 
based on the SRS $X_1,X_2,...,X_n$ drawn from $exp(\lambda)$ population is given by
\[T_{SRS}=\left(1-{\frac{Q}{W}}\right)^{n-1}_{+}\]
where, $W={\sum_{i=1}^{n}}X_i=n{\bar{X}}$ and $(d)_{+}=max(d,0)$.
First we replace $e^{-\frac{Q}{\lambda}}$ by $T_{SRS}$ in eq(\ref{FOCExp}). Also, the UMVUE of $\lambda$, appearing in the coefficients of $Q^{m-j-1}, \; \forall 0\leq j \leq m-1$, in eq(\ref{FOCExp}), is given by $\bar{X}$. Replacing the above two estimators in place of their corresponding estimands in FOC, the second estimating equation is obtained as
\begin{equation}
{\sum_{j=0}^{m-1}}(-1)^j  \binom{n-1}{m-j-1}\left(\frac{Q}{W}\right)^{m-j-1} =\gamma_{m}\left(1-{\frac{Q}{W}}\right)^{n-1}_{+} 
\label{eeq2}
\end{equation}
Let the solution of the above estimating equation be denoted by $\hat{Q}^*_2$. This is also a plug-in estimator and we compare it with  $\hat{Q}_1^*$ in terms of bias and MSE, as provided in  section 4.

\vspace{12pt}
\subsubsection{Estimating equation based on order statistics}
In case the full data is not available, it is more likely that the seller would be able to recall the worst day or the best day in terms of demand. The worst day demand would be represented by the smallest order statistic $X_{(1)}$, and the best day would be counted as the largest order statistic $X_{(n)}$. In general, if the observation on the $i^{th}$ smallest demand out of a sample of size $n$ is available, \emph{viz.} $X_{(i)}$, then we may consider the following two possible ways of estimating the optimal order quantity. One is to use estimating equation obtained by replacing $\lambda$ with its unbiased estimator based on $X_{(i)}$ in both sides of the eq.~(\ref{FOCExp}) \citep[see][and references therein]{sengupta2006unbiased}. Another estimating equation can be obtained by replacing $\lambda$ by its unbiased estimator using $X_{(i)}$ in the left hand side expression of eq.~(\ref{FOCExp}) and $X_{(i)}$ based unbiased estimator of $\bar{F}(Q)$ in the right hand side of the same equation.

%
\citet{sinha2006} provided an unbiased estimator of ${\bar{F}}(Q)$ based on $X_{(i+1)}$, which is  as follows:
\begin{equation}
h_{i+1}(Z_{i+1})=\sum_{j_1=0}^\infty \sum_{j_2=0}^\infty\ldots \sum_{j_i=0}^\infty d_{j_1j_2\ldots j_i}I(Z_{i+1}>\alpha_1^{j_1}\alpha_2^{j_2} \ldots \alpha_i^{j_i}Q)
\end{equation}
where $\displaystyle Z_i=(n-i+1)X_{(i)}$,  $\alpha_k=\frac{n-i+k}{n-i}, \; k=1,2 \ldots i$, $\displaystyle d_{j_1j_2\ldots j_i}=\frac{(-1)^{\sum_1 j_i}}{\binom{n}{i}}\times$ $\displaystyle \prod_{k=1}^{i}\left[\frac{\binom{i}{k}}{\alpha_k}\right]^{j_k}$, $\sum_1$ extending over all even suffixes of $j$. 
Further, an unbiased estimator of $\lambda$ based on $X_{(i)}$ is given by $\hat{\lambda}=\frac{X_{(i)}}{a_i}$, where $a_i=\displaystyle \sum_{j=1}^i\frac{1}{n-j+1}$. The first approach to estimate $Q$ would be through the estimating equation obtained by replacing $\lambda$ with $\hat{\lambda}$ in eq.~\ref{FOCExp}. Let the corresponding solution be denoted by $\hat{Q}^*_{(1)}$. The second approach for estimating $Q^*$ is to replace $\bar{F}(Q)$ by $h_{i}(Z_{i})$ and $\lambda$ by $\hat{\lambda}$ in eq.~(\ref{FOCExp}). We denote the corresponding estimator by  $\hat{Q}^*_{(2)}$. 

In the case of only the best day demand data being available, the largest order statistic $X_{(n)}$ is observed. However, due to high complexity of computation we don't investigate this case in details.

\section{Simulation}
In this section we present simulation studies for exponential demand in order to estimate the optimum order quantities from the estimating equations discussed above. Let us consider  standard exponential demand distribution (exp($1$)) and observe performances of the estimated optimal order quantities corresponding to different estimating equations over different values of $m$ and ${\alpha_m}$. In this simulation we draw samples of size $n$, ($=10,50, 100,500, 1000,5000,10000$). For a given sample size $n$, the estimated optimal quantities are determined from each of the estimating equations described in the previous section. This procedure is repeated $1000$
times. We compute the bias in the estimated optimal order quantities by the average of ($\hat{Q}^*-Q^*$) over these $1000$ repetitions, where $Q^*$ is true optimal order quantity and $\hat{Q}^*$ is its estimate. The tables below report bias and MSE of different estimators of the optimal order quantity.

Full sample bias and MSE of $\hat{Q}_1^*$ are reported in tables \ref{q1starhat1}-\ref{q1starhat3}, respectively and the same for $\hat{Q}_2^*$ are given in tables \ref{q2starhat1}-\ref{q2starhat3}. It can be observed from the figures that the bias and MSEs of these two estimators of $Q^*$ are comparable. Also, none of the two estimators  uniformly outperforms the other in terms of absolute bias or MSE across the given degrees of importance of loss, $m$ , in either small or large sample cases.

Comparing the bias and MSE of the estimators of $Q^*$ based on $2^{nd}$ order statistic (\emph{viz.} $\hat{Q}_{(1)}^*$ and $\hat{Q}_{(2)}^*$) given in tables Tab.\ref{qordr1starhat1}-\ref{qordr1starhat3} and Tab.\ref{qordr2starhat1}-\ref{qordr2starhat3} respectively, similar observations as in the full sample case, can be made. However, in this case, the margin in bias and/or MSE for certain $(\alpha_1,m,n)$ are much larger. For example, MSE of $\hat{Q}_{(1)}^*$ is quite smaller than that of $\hat{Q}_{(2)}^*$ for $m=50$ over all the considered sample sizes, when $\alpha_1=2$, whereas similar margins could be observed favoring $\hat{Q}_{(2)}^*$ in case of $\alpha_1=1$ and $m=10$, for all sample sizes except 100 and 1000. Thus, neither of the two estimators outperform the other.  

\section{Conclusion}
Contributions of our study in this paper are 2-fold. First we have proposed a generalization of the standard news vendor problem assuming random demand and higher degree of shortage and excess loss. In particular, we have developed a symmetric generalized news vendor cost structure using power losses of same degree for both of shortage and excess inventory. We have presented the method to determine the optimal order quantity. In particular, we have presented the analytical expression of the optimal order quantity for uniform demand. In case of exponential demand,  determination of optimal order quantity requires finding zeroes of a transcendental equation. We have proven the existence of real roots of the equation, ensuring that there exists a realistic solution to the proposed general model. Our second contribution is to provide different estimators of the optimal order quantity. We have provided estimators of the optimal order quantity using (i) full sample and (ii) broken sample data on demand. Whereas, analytical form of the estimator and its properties are easy to verify for uniform demand distribution, it is difficult for exponential distribution. We have presented a simulation study to compare different estimators. In this paper, we have presented estimators based on full sample as well as broken sample like single order statistic. Finally, we have provided a simulation study to gauge the performance of the proposed estimators in terms of bias and MSE. 

A natural extension of this work could be to consider asymmetric shortage and excess losses. However, asymmetric power type shortage and excess losses would result in different dimensions of shortage and excess costs making them incomparable. Unpublished manuscript by \citet{baraiya2019} proposes an inventory model for such a case.

\section*{Acknowledgement} 
$3^{rd}$ author remains deeply indebted to Late Prof. S Sengupta and Prof. Bikas K Sinha for their intriguing discussions and guidance on broken sample estimation. Research of $1^{st}$ author is funded by Indian Institute of Management Indore. The authors acknowledge gratefully the valuable comments by the participants of the SMDTDS-2020 conference organized by IAPQR at Kolkata.

\bibliographystyle{plainnat}
\bibliography{EOQRef}

\newpage
\setcounter{section}{0}
\renewcommand{\thesection}{\Roman{section}}
\appendix

\section{Tables}
\begin{table}[h!]
	\renewcommand{\thetable}{\arabic{table}a}
	\centering
	\caption{Bias (in $1^{st}$ row) and MSE (in $2^{nd}$ row) of $\hat{Q}^*_{1}$ for $C_s=2C_e$}
	\begin{tabular}{rrrrrrrr}
		\hline
		m & $n=$10 & $n=$50 & $n=$100 & $n=$500 & $n=$1000 & $n=$5000 & $n=$10000 \\ 
		\hline
		2 & 0.0052 & -0.0070 & -0.0010 & 0.0011 & -0.0007 & 0.0004 & -0.0006 \\ 
		& 0.0597 & 0.0113 & 0.0060 & 0.0012 & 0.0005 & 0.0001 & 0.0001 \\ 
		3 & 0.6206 & 0.5937 & 0.6070 & 0.6116 & 0.6076 & 0.6099 & 0.6079 \\ 
		& 0.6775 & 0.4074 & 0.3976 & 0.3800 & 0.3718 & 0.3726 & 0.3698 \\ 
		4 & 13.1606 & 12.9312 & 13.0446 & 13.0836 & 13.0495 & 13.0695 & 13.0523 \\ 
		& 194.3676 & 171.1996 & 172.2751 & 171.6106 & 170.4826 & 170.8554 & 170.3837 \\ 
		5 & -0.0927 & -0.1180 & -0.1055 & -0.1012 & -0.1050 & -0.1028 & -0.1047 \\ 
		& 0.2664 & 0.0624 & 0.0369 & 0.0155 & 0.0134 & 0.0111 & 0.0112 \\ 
		10 & 3.0587 & 2.9607 & 3.0092 & 3.0258 & 3.0112 & 3.0198 & 3.0124 \\ 
		& 13.2180 & 9.4926 & 9.4406 & 9.2340 & 9.1028 & 9.1270 & 9.0786 \\ 
		20 & -4.4660 & -4.4904 & -4.4784 & -4.4742 & -4.4778 & -4.4757 & -4.4776 \\ 
		& 20.1848 & 20.2090 & 20.0797 & 20.0235 & 20.0533 & 20.0325 & 20.0487 \\ 
		50 & -10.9240 & -10.9797 & -10.9522 & -10.9427 & -10.9510 & -10.9461 & -10.9503 \\ 
		& 120.5828 & 120.7890 & 120.0743 & 119.7677 & 119.9350 & 119.8197 & 119.9100 \\ 
		\hline
	\end{tabular}
	\label{q1starhat1}
\end{table}

\begin{table}
	\addtocounter{table}{-1}
	\renewcommand{\thetable}{\arabic{table}b}
	\centering
	\caption{Bias (in $1^{st}$ row) and MSE (in $2^{nd}$ row) of $\hat{Q}^*_{1}$ for $C_s = C_e$}
	\begin{tabular}{rrrrrrrr}
		\hline
		m & $n=$10 & $n=$50 & $n=$100 & $n=$500 & $n=$1000 & $n=$5000 & $n=$10000 \\ 
		\hline
		2 & 0.0978 & 0.0805 & 0.0891 & 0.0920 & 0.0895 & 0.0910 & 0.0897 \\ 
		& 0.1299 & 0.0291 & 0.0200 & 0.0109 & 0.0091 & 0.0085 & 0.0082 \\ 
		3 & 1.8820 & 1.8318 & 1.8566 & 1.8651 & 1.8577 & 1.8621 & 1.8583 \\ 
		& 4.5529 & 3.5459 & 3.5480 & 3.4993 & 3.4603 & 3.4694 & 3.4543 \\ 
		4 & -0.5893 & -0.6052 & -0.5973 & -0.5946 & -0.5970 & -0.5956 & -0.5968 \\ 
		& 0.4485 & 0.3853 & 0.3669 & 0.3556 & 0.3573 & 0.3549 & 0.3563 \\ 
		5 & 0.0128 & -0.0172 & -0.0024 & 0.0027 & -0.0017 & 0.0009 & -0.0014 \\ 
		& 0.3614 & 0.0683 & 0.0361 & 0.0074 & 0.0033 & 0.0007 & 0.0004 \\ 
		10 & 11.3833 & 11.1515 & 11.2661 & 11.3055 & 11.2711 & 11.2913 & 11.2739 \\ 
		& 151.2083 & 128.4259 & 129.0847 & 128.2545 & 127.2342 & 127.5375 & 127.1223 \\ 
		20 & -4.3518 & -4.3806 & -4.3664 & -4.3615 & -4.3657 & -4.3632 & -4.3654 \\ 
		& 19.2711 & 19.2521 & 19.0983 & 19.0292 & 19.0627 & 19.0385 & 19.0570 \\ 
		50 & -8.2308 & -8.3314 & -8.2817 & -8.2646 & -8.2795 & -8.2708 & -8.2783 \\ 
		& 71.8199 & 70.1792 & 68.9928 & 68.3861 & 68.5878 & 68.4136 & 68.5345 \\ 
		\hline
	\end{tabular}
	\label{q1starhat2}
\end{table}

\begin{table}
	\addtocounter{table}{-1}
	\renewcommand{\thetable}{\arabic{table}c}
	\centering
	\caption{Bias (in $1^{st}$ row) and MSE (in $2^{nd}$ row) of $\hat{Q}^*_{1}$ for $C_s =0.5 C_e$}
	\begin{tabular}{rrrrrrrr}
		\hline
		$m$ & $n=$10 & $n=$50 & $n=$100 & $n=$500 & $n=$1000 & $n=$5000 & $n=$10000 \\ 
		\hline
		2 & 0.1296 & 0.1074 & 0.1184 & 0.1222 & 0.1189 & 0.1208 & 0.1191 \\ 
		& 0.2148 & 0.0488 & 0.0338 & 0.0189 & 0.0159 & 0.0150 & 0.0144 \\ 
		3 & 4.5171 & 4.4217 & 4.4689 & 4.4851 & 4.4709 & 4.4793 & 4.4721 \\ 
		& 24.0654 & 20.2408 & 20.3366 & 20.1907 & 20.0227 & 20.0712 & 20.0034 \\ 
		4 & -0.5046 & -0.5252 & -0.5150 & -0.5115 & -0.5146 & -0.5128 & -0.5143 \\ 
		& 0.4257 & 0.3080 & 0.2823 & 0.2651 & 0.2664 & 0.2633 & 0.2647 \\ 
		5 & 1.2617 & 1.2088 & 1.2349 & 1.2439 & 1.2361 & 1.2407 & 1.2367 \\ 
		& 2.7166 & 1.6729 & 1.6374 & 1.5702 & 1.5381 & 1.5416 & 1.5306 \\ 
		10 & -2.2265 & -2.2468 & -2.2367 & -2.2333 & -2.2363 & -2.2345 & -2.2361 \\ 
		& 5.1226 & 5.0790 & 5.0195 & 4.9909 & 5.0025 & 4.9934 & 5.0001 \\ 
		20 & -4.2357 & -4.2689 & -4.2525 & -4.2468 & -4.2517 & -4.2488 & -4.2513 \\ 
		& 18.3848 & 18.3069 & 18.1277 & 18.0444 & 18.0814 & 18.0536 & 18.0743 \\ 
		50 & 0.1000 & -0.1344 & -0.0185 & 0.0213 & -0.0135 & 0.0069 & -0.0107 \\ 
		& 22.1078 & 4.1763 & 2.2061 & 0.4496 & 0.2016 & 0.0449 & 0.0223 \\ 
		\hline
	\end{tabular}
	\label{q1starhat3}
\end{table}
\begin{table}
	\renewcommand{\thetable}{\arabic{table}a}
	\centering
	\caption{Bias (in $1^{st}$ row) and MSE (in $2^{nd}$ row) of $\hat{Q}^*_{2}$ for $C_e =2 C_s$\vspace{12pt}}
	\begin{tabular}{llllllll}
		\hline
		m & $n=$10 & $n=$50 & $n=$100 & $n=$500 & $n=$1000 & $n=$5000 & $n=$10000 \\ 
		\hline		
		2 & -0.0040 & -0.0062 & 0.0024 & 0.0000 & -0.0003 & -0.0001 & -0.0002 \\ 
		& 0.0566 & 0.0110 & 0.0060 & 0.0012 & 0.0006 & 0.0001 & 0.0001 \\ 
		3 & 0.6517 & 0.6058 & 0.6197 & 0.6102 & 0.6091 & 0.6090 & 0.6087 \\ 
		& 0.7187 & 0.4211 & 0.4138 & 0.3785 & 0.3740 & 0.3715 & 0.3709 \\ 
		4 & 13.4388 & 13.1458 & 13.2482 & 13.1041 & 13.0796 & 13.0654 & 13.0613 \\ 
		& 201.9324 & 176.7958 & 177.6949 & 172.1608 & 171.2964 & 170.7480 & 170.6197 \\ 
		5 & -0.0795 & -0.1103 & -0.0954 & -0.1028 & -0.1037 & -0.1036 & -0.1039 \\ 
		& 0.2606 & 0.0597 & 0.0352 & 0.0160 & 0.0134 & 0.0113 & 0.0111 \\ 
		10 & 3.1096 & 2.9906 & 3.0482 & 3.0195 & 3.0161 & 3.0164 & 3.0154 \\ 
		& 13.4791 & 9.6562 & 9.6832 & 9.1984 & 9.1368 & 9.1061 & 9.0969 \\ 
		20 & -4.4347 & -4.4792 & -4.4667 & -4.4754 & -4.4765 & -4.4765 & -4.4768 \\ 
		& 19.9084 & 20.1073 & 19.9757 & 20.0342 & 20.0411 & 20.0399 & 20.0419 \\ 
		50 & -10.9395 & -10.9700 & -10.9336 & -10.9470 & -10.9486 & -10.9481 & -10.9486 \\ 
		& 120.8739 & 120.5714 & 119.6690 & 119.8620 & 119.8843 & 119.8641 & 119.8738 \\ 
		\hline
	\end{tabular}
	\label{q2starhat1}
\end{table}

\begin{table}
	\addtocounter{table}{-1}
	\renewcommand{\thetable}{\arabic{table}b}
	\centering
	\caption{Bias (in $1^{st}$ row) and MSE (in $2^{nd}$ row) of $\hat{Q}^*_{2}$ for $C_s =2 C_e$\vspace{12pt}}
	\begin{tabular}{llllllll}
		\hline
		m & $n=$10 & $n=$50 & $n=$100 & $n=$500 & $n=$1000 & $n=$5000 & $n=$10000 \\ 
		\hline		
		2 & 0.1491 & 0.0942 & 0.1002 & 0.0918 & 0.0907 & 0.0904 & 0.0902 \\ 
		& 0.1502 & 0.0314 & 0.0223 & 0.0109 & 0.0095 & 0.0084 & 0.0083 \\ 
		3 & 1.9619 & 1.8564 & 1.8812 & 1.8628 & 1.8606 & 1.8604 & 1.8599 \\ 
		& 4.8798 & 3.6337 & 3.6416 & 3.4912 & 3.4723 & 3.4630 & 3.4602 \\ 
		4 & -0.5811 & -0.6003 & -0.5910 & -0.5956 & -0.5962 & -0.5962 & -0.5963 \\ 
		& 0.4374 & 0.3791 & 0.3595 & 0.3569 & 0.3565 & 0.3556 & 0.3557 \\ 
		5 & 0.0309 & -0.0070 & 0.0101 & 0.0009 & -0.0002 & -0.0002 & -0.0004 \\ 
		& 0.3582 & 0.0668 & 0.0368 & 0.0076 & 0.0038 & 0.0007 & 0.0004 \\ 
		10 & 11.5039 & 11.2222 & 11.3586 & 11.2907 & 11.2825 & 11.2832 & 11.2810 \\ 
		& 153.6697 & 129.9282 & 131.2102 & 127.9323 & 127.5194 & 127.3534 & 127.2837 \\ 
		20 & -4.3348 & -4.3713 & -4.3546 & -4.3633 & -4.3643 & -4.3642 & -4.3645 \\ 
		& 19.1194 & 19.1696 & 18.9964 & 19.0449 & 19.0506 & 19.0472 & 19.0494 \\ 
		50 & -8.3287 & -8.3343 & -8.2585 & -8.2744 & -8.2763 & -8.2746 & -8.2754 \\ 
		& 73.1996 & 70.2044 & 68.6129 & 68.5503 & 68.5387 & 68.4776 & 68.4869 \\ 
		\hline
	\end{tabular}
	\label{q2starhat2}
\end{table}

\begin{table}
	\addtocounter{table}{-1}
	\renewcommand{\thetable}{\arabic{table}c}
	\centering
	\caption{Bias (in $1^{st}$ row) and MSE (in $2^{nd}$ row) of $\hat{Q}^*_{2}$ for $C_s =2 C_e$\vspace{12pt}}
	\begin{tabular}{llllllll}
		\hline
		m & $n=$10 & $n=$50 & $n=$100 & $n=$500 & $n=$1000 & $n=$5000 & $n=$10000 \\ 
		\hline		
		2 & 0.1288 & 0.1119 & 0.1261 & 0.1205 & 0.1198 & 0.1200 & 0.1198 \\ 
		& 0.2084 & 0.0489 & 0.0359 & 0.0187 & 0.0164 & 0.0148 & 0.0146 \\ 
		3 & 4.7220 & 4.5005 & 4.5335 & 4.4846 & 4.4784 & 4.4765 & 4.4753 \\ 
		& 26.0939 & 20.9409 & 20.9274 & 20.1881 & 20.0944 & 20.0462 & 20.0321 \\ 
		4 & -0.4644 & -0.5131 & -0.5038 & -0.5122 & -0.5132 & -0.5134 & -0.5137 \\ 
		& 0.3920 & 0.2951 & 0.2713 & 0.2660 & 0.2652 & 0.2639 & 0.2640 \\ 
		5 & 1.2892 & 1.2249 & 1.2560 & 1.2405 & 1.2387 & 1.2388 & 1.2383 \\ 
		& 2.7712 & 1.7080 & 1.6916 & 1.5625 & 1.5460 & 1.5369 & 1.5347 \\ 
		10 & -2.2054 & -2.2386 & -2.2276 & -2.2344 & -2.2352 & -2.2352 & -2.2354 \\ 
		& 5.0296 & 5.0418 & 4.9791 & 4.9959 & 4.9978 & 4.9965 & 4.9973 \\ 
		20 & -4.2275 & -4.2598 & -4.2397 & -4.2490 & -4.2502 & -4.2500 & -4.2503 \\ 
		& 18.3061 & 18.2278 & 18.0200 & 18.0635 & 18.0684 & 18.0636 & 18.0658 \\ 
		50 & 0.0614 & -0.1994 & -0.0146 & -0.0146 & -0.0126 & -0.0034 & -0.0046 \\ 
		& 21.3335 & 4.0419 & 2.2135 & 0.4609 & 0.2292 & 0.0443 & 0.0239 \\ 
		\hline
	\end{tabular}
	\label{q2starhat3}
\end{table}
\begin{table}
	\renewcommand{\thetable}{\arabic{table}a}
	\centering
	\caption{Bias (in $1^{st}$ row) and MSE (in $2^{nd}$ row) of $\hat{Q}^*_{(1)}$ for $C_s =2 C_e$\vspace{12pt}}
	\begin{tabular}{llllllll}
		m & $n=$10 & $n=$50 & $n=$100 & $n=$500 & $n=$1000 & $n=$5000 & $n=$10000 \\ 
		\hline
		2 & 0.0311 & 0.0057 & 0.0044 & 0.0118 & 0.0130 & 0.0042 & 0.0257 \\ 
		& 0.3344 & 0.2931 & 0.2813 & 0.3175 & 0.2928 & 0.3210 & 0.3657 \\ 
		3 & 0.6780 & 0.6219 & 0.6189 & 0.6353 & 0.6379 & 0.6185 & 0.6660 \\ 
		& 2.0924 & 1.8219 & 1.7604 & 1.9575 & 1.8396 & 1.9542 & 2.2309 \\ 
		4 & 13.6486 & 13.1712 & 13.1460 & 13.2854 & 13.3079 & 13.1422 & 13.5463 \\ 
		& 304.4965 & 277.3833 & 272.5336 & 289.0052 & 280.8245 & 286.5041 & 312.9112 \\ 
		5 & -0.0388 & -0.0915 & -0.0943 & -0.0789 & -0.0765 & -0.0947 & -0.0501 \\ 
		& 1.4415 & 1.2740 & 1.2235 & 1.3767 & 1.2693 & 1.3950 & 1.5789 \\ 
		10 & 3.2672 & 3.0633 & 3.0525 & 3.1120 & 3.1216 & 3.0508 & 3.2235 \\ 
		& 32.2458 & 28.3435 & 27.5137 & 30.2142 & 28.6722 & 30.0715 & 34.0053 \\ 
		20 & -4.4141 & -4.4649 & -4.4676 & -4.4528 & -4.4504 & -4.4680 & -4.4250 \\ 
		& 20.8213 & 21.1102 & 21.0868 & 21.0992 & 20.9786 & 21.2496 & 21.0440 \\ 
		50 & -10.8054 & -10.9214 & -10.9275 & -10.8936 & -10.8882 & -10.9285 & -10.8303 \\ 
		& 123.7353 & 125.4102 & 125.2968 & 125.3128 & 124.6758 & 126.1481 & 124.9338 \\ 
		\hline
	\end{tabular}
	\label{qordr1starhat1}
\end{table}

\begin{table}
	\addtocounter{table}{-1}
	\renewcommand{\thetable}{\arabic{table}b}
	\centering
	\caption{Bias (in $1^{st}$ row) and MSE (in $2^{nd}$ row) of $\hat{Q}^*_{(1)}$ for $C_s = C_e$\vspace{12pt}}
	\begin{tabular}{llllllll}
		m & $n=$10 & $n=$50 & $n=$100 & $n=$500 & $n=$1000 & $n=$5000 & $n=$10000 \\ 
		\hline
		2 & 0.1346 & 0.0986 & 0.0967 & 0.1073 & 0.1089 & 0.0965 & 0.1269 \\ 
		& 0.6903 & 0.6005 & 0.5763 & 0.6512 & 0.6016 & 0.6563 & 0.7519 \\ 
		3 & 1.9886 & 1.8843 & 1.8788 & 1.9092 & 1.9141 & 1.8779 & 1.9663 \\ 
		& 9.6012 & 8.5136 & 8.2929 & 9.0191 & 8.6186 & 8.9619 & 10.0476 \\ 
		4 & -0.5556 & -0.5886 & -0.5903 & -0.5807 & -0.5791 & -0.5906 & -0.5626 \\ 
		& 0.8739 & 0.8432 & 0.8253 & 0.8751 & 0.8314 & 0.8929 & 0.9354 \\ 
		5 & 0.0765 & 0.0142 & 0.0109 & 0.0291 & 0.0320 & 0.0104 & 0.0632 \\ 
		& 2.0236 & 1.7737 & 1.7021 & 1.9211 & 1.7715 & 1.9423 & 2.2128 \\ 
		10 & 11.8767 & 11.3942 & 11.3687 & 11.5096 & 11.5323 & 11.3648 & 11.7733 \\ 
		& 261.8440 & 235.9926 & 231.1351 & 247.4247 & 238.9786 & 245.4252 & 270.8390 \\ 
		20 & -4.2906 & -4.3505 & -4.3536 & -4.3362 & -4.3333 & -4.3541 & -4.3034 \\ 
		& 20.2683 & 20.5605 & 20.5222 & 20.5714 & 20.4089 & 20.7477 & 20.5546 \\ 
		50 & -8.0167 & -8.2261 & -8.2372 & -8.1760 & -8.1662 & -8.2389 & -8.0616 \\ 
		& 87.0193 & 87.6663 & 87.0425 & 88.5002 & 86.6497 & 89.7789 & 89.8955 \\ 
		\hline
	\end{tabular}
	\label{qordr1starhat2}
\end{table}
\begin{table}
	\addtocounter{table}{-1}
	\renewcommand{\thetable}{\arabic{table}c}
	\centering
	\caption{Bias (in $1^{st}$ row) and MSE (in $2^{nd}$ row) of $\hat{Q}^*_{(1)}$ for $C_s =0.5 C_e$\vspace{12pt}}
	\begin{tabular}{llllllll}
		m & $n=$10 & $n=$50 & $n=$100 & $n=$500 & $n=$1000 & $n=$5000 & $n=$10000 \\ 
		\hline
		2 & 0.1768 & 0.1306 & 0.1282 & 0.1417 & 0.1439 & 0.1278 & 0.1669 \\ 
		& 1.1370 & 0.9889 & 0.9491 & 1.0724 & 0.9909 & 1.0807 & 1.2383 \\ 
		3 & 4.7201 & 4.5216 & 4.5111 & 4.5691 & 4.5784 & 4.5095 & 4.6776 \\ 
		& 42.7250 & 38.4152 & 37.5964 & 40.3345 & 38.9017 & 40.0158 & 44.2617 \\ 
		4 & -0.4607 & -0.5036 & -0.5059 & -0.4933 & -0.4913 & -0.5062 & -0.4699 \\ 
		& 1.1676 & 1.0934 & 1.0618 & 1.1526 & 1.0797 & 1.1759 & 1.2667 \\ 
		5 & 1.3742 & 1.2641 & 1.2583 & 1.2905 & 1.2956 & 1.2574 & 1.3506 \\ 
		& 8.1699 & 7.1191 & 6.8820 & 7.6434 & 7.1903 & 7.6275 & 8.7005 \\ 
		10 & -2.1833 & -2.2255 & -2.2278 & -2.2154 & -2.2135 & -2.2281 & -2.1924 \\ 
		& 5.6909 & 5.7651 & 5.7423 & 5.7875 & 5.7101 & 5.8538 & 5.8179 \\ 
		20 & -4.1650 & -4.2341 & -4.2378 & -4.2176 & -4.2143 & -4.2383 & -4.1798 \\ 
		& 19.8266 & 20.1071 & 20.0502 & 20.1477 & 19.9362 & 20.3501 & 20.1851 \\ 
		50 & 0.5986 & 0.1109 & 0.0852 & 0.2276 & 0.2505 & 0.0812 & 0.4941 \\ 
		& 123.7725 & 108.4855 & 104.1096 & 117.5048 & 108.3515 & 118.8006 & 135.3470 \\ 
		\hline
	\end{tabular}
	\label{qordr1starhat3}
\end{table}


\begin{table}
	\renewcommand{\thetable}{\arabic{table}a}
	\centering
	\caption{Bias (in $1^{st}$ row) and MSE (in $2^{nd}$ row) of $\hat{Q}^*_{(2)}$ for $C_s =2 C_e$\vspace{12pt}}
	\begin{tabular}{llllllll}
		\hline
		m & $n=$10 & $n=$50 & $n=$100 & $n=$500 & $n=$1000 & $n=$5000 & $n=$10000 \\ 
		\hline
		2 & -0.0625 & -0.1219 & -0.0799 & -0.1007 & -0.1041 & -0.1025 & -0.1056 \\ 
		& 0.2761 & 0.2108 & 0.2245 & 0.2384 & 0.2546 & 0.2371 & 0.2576 \\ 
		3 & 0.5977 & 0.4855 & 0.5937 & 0.5472 & 0.5395 & 0.5451 & 0.5376 \\ 
		& 1.9156 & 1.4015 & 1.6598 & 1.6739 & 1.7599 & 1.6658 & 1.7779 \\ 
		4 & 13.2132 & 12.6096 & 13.6049 & 13.2154 & 13.1489 & 13.1931 & 13.1255 \\ 
		& 291.3362 & 251.0958 & 288.8477 & 284.1045 & 289.9123 & 282.9913 & 290.7649 \\ 
		5 & -0.1042 & -0.1701 & -0.0614 & -0.1039 & -0.1112 & -0.1064 & -0.1138 \\ 
		& 1.4026 & 1.1268 & 1.2407 & 1.3157 & 1.4074 & 1.3100 & 1.4255 \\ 
		10 & 3.0143 & 2.7592 & 3.1798 & 3.0152 & 2.9871 & 3.0058 & 2.9772 \\ 
		& 29.9356 & 24.0600 & 28.6405 & 28.6393 & 29.8210 & 28.4890 & 30.0241 \\ 
		20 & -4.5136 & -4.5327 & -4.4134 & -4.4535 & -4.4608 & -4.4556 & -4.4627 \\ 
		& 21.6035 & 21.5750 & 20.6595 & 21.0824 & 21.2326 & 21.0951 & 21.2678 \\ 
		50 & -11.1292 & -11.2668 & -11.0399 & -11.1287 & -11.1439 & -11.1338 & -11.1492 \\ 
		& 129.9304 & 131.7310 & 127.2745 & 129.5401 & 130.2713 & 129.6265 & 130.4667 \\ \hline 
	\end{tabular}
	\label{qordr2starhat1}
\end{table}

\begin{table}
	\addtocounter{table}{-1}
	\renewcommand{\thetable}{\arabic{table}b}
	\centering
	\caption{Bias (in $1^{st}$ row) and MSE (in $2^{nd}$ row) of $\hat{Q}^*_{(2)}$ for $C_s = C_e$\vspace{12pt}}
	\begin{tabular}{llllllll}
		\hline
		m & $n=$10 & $n=$50 & $n=$100 & $n=$500 & $n=$1000 & $n=$5000 & $n=$10000 \\ 
		\hline
		2 & 0.1857 & 0.2157 & 0.3067 & 0.2795 & 0.2740 & 0.2783 & 0.2725 \\ 
		& 0.8032 & 0.7401 & 0.8810 & 0.9173 & 0.9725 & 0.9135 & 0.9837 \\ 
		3 & 2.0321 & 1.8945 & 2.1214 & 2.0326 & 2.0174 & 2.0275 & 2.0121 \\ 
		& 10.2007 & 8.3781 & 9.8960 & 9.8237 & 10.1555 & 9.7758 & 10.2103 \\ 
		4 & -0.5965 & -0.6378 & -0.5697 & -0.5963 & -0.6009 & -0.5979 & -0.6025 \\ 
		& 0.9022 & 0.8377 & 0.8101 & 0.8679 & 0.9087 & 0.8672 & 0.9175 \\ 
		5 & -0.2011 & -0.2469 & -0.1175 & -0.1674 & -0.1741 & -0.1685 & -0.1766 \\ 
		& 1.5988 & 1.3269 & 1.4608 & 1.5476 & 1.6572 & 1.5434 & 1.6790 \\ 
		10 & 11.2783 & 10.6747 & 11.6700 & 11.2805 & 11.2140 & 11.2581 & 11.1905 \\ 
		& 243.9466 & 206.0419 & 239.9422 & 236.7064 & 242.7717 & 235.6797 & 243.7150 \\ 
		20 & -4.4437 & -4.4980 & -4.3694 & -4.4175 & -4.4258 & -4.4206 & -4.4287 \\ 
		& 21.3903 & 21.5558 & 20.5983 & 21.1017 & 21.2846 & 21.1207 & 21.3313 \\ 
		50 & -8.4431 & -8.6981 & -8.2775 & -8.4421 & -8.4703 & -8.4516 & -8.4802 \\ 
		& 92.1351 & 92.1044 & 87.0470 & 90.8175 & 92.6434 & 90.8839 & 93.0735 \\ 
		\hline
	\end{tabular}
	\label{qordr2starhat2}
\end{table}

\begin{table}[h]
	\addtocounter{table}{-1}
	\renewcommand{\thetable}{\arabic{table}c}
	\centering
	\caption{Bias (in $1^{st}$ row) and MSE (in $2^{nd}$ row) of $\hat{Q}^*_{(2)}$ for $C_s = 0.5 C_e$\vspace{12pt}}
	\begin{tabular}{llllllll}
		\hline
		$m$ & $n=$10 & $n=$50 & $n=$100 & $n=$500 & $n=$1000 & $n=$5000 & $n=$10000 \\ 
		\hline
		2 & 0.0783 & -0.0096 & 0.0751 & 0.0353 & 0.0286 & 0.0322 & 0.0261 \\ 
		& 1.0075 & 0.7557 & 0.8501 & 0.8859 & 0.9455 & 0.8799 & 0.9566 \\ 
		3 & 4.6372 & 4.3821 & 4.8027 & 4.6381 & 4.6100 & 4.6286 & 4.6001 \\ 
		& 42.3529 & 35.6494 & 41.5951 & 41.0597 & 42.1501 & 40.8788 & 42.3211 \\ 
		4 & -0.4795 & -0.4729 & -0.3671 & -0.4008 & -0.4066 & -0.4019 & -0.4084 \\ 
		& 1.2028 & 1.0670 & 1.0989 & 1.1834 & 1.2596 & 1.1810 & 1.2757 \\ 
		5 & 1.2377 & 1.1001 & 1.3270 & 1.2382 & 1.2230 & 1.2331 & 1.2177 \\ 
		& 7.6032 & 5.9993 & 7.1566 & 7.2254 & 7.5813 & 7.1856 & 7.6446 \\ 
		10 & -2.1798 & -2.2379 & -2.1481 & -2.1815 & -2.1869 & -2.1827 & -2.1889 \\ 
		& 5.7243 & 5.7721 & 5.4737 & 5.6682 & 5.7558 & 5.6704 & 5.7768 \\ 
		20 & -4.4662 & -4.6036 & -4.4881 & -4.5374 & -4.5444 & -4.5388 & -4.5471 \\ 
		& 21.8757 & 22.6154 & 21.7307 & 22.2611 & 22.4430 & 22.2681 & 22.4902 \\ 
		50 & -0.1642 & -0.7678 & 0.2275 & -0.1620 & -0.2285 & -0.1843 & -0.2519 \\ 
		& 116.7738 & 92.6822 & 103.8056 & 109.4833 & 117.0710 & 108.9685 & 118.5504 \\ 
		\hline	
	\end{tabular}	
	\label{qordr2starhat3}
\end{table}

\clearpage
\renewcommand{\thesection}{\Roman{section}}
\section{Figures}
\begin{figure}[htp]
	\includegraphics[scale=0.45]{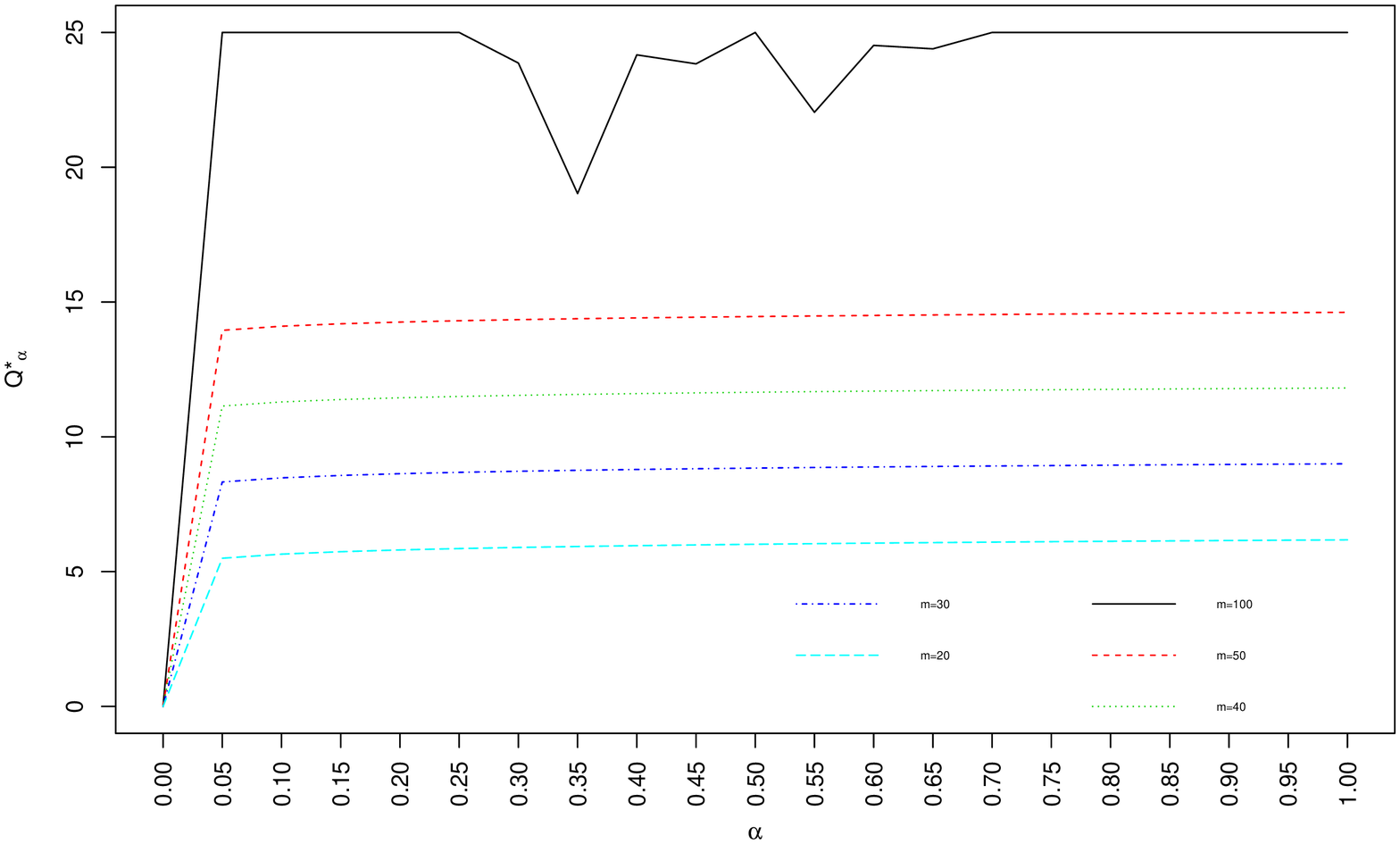}\\
	\includegraphics[scale=0.45]{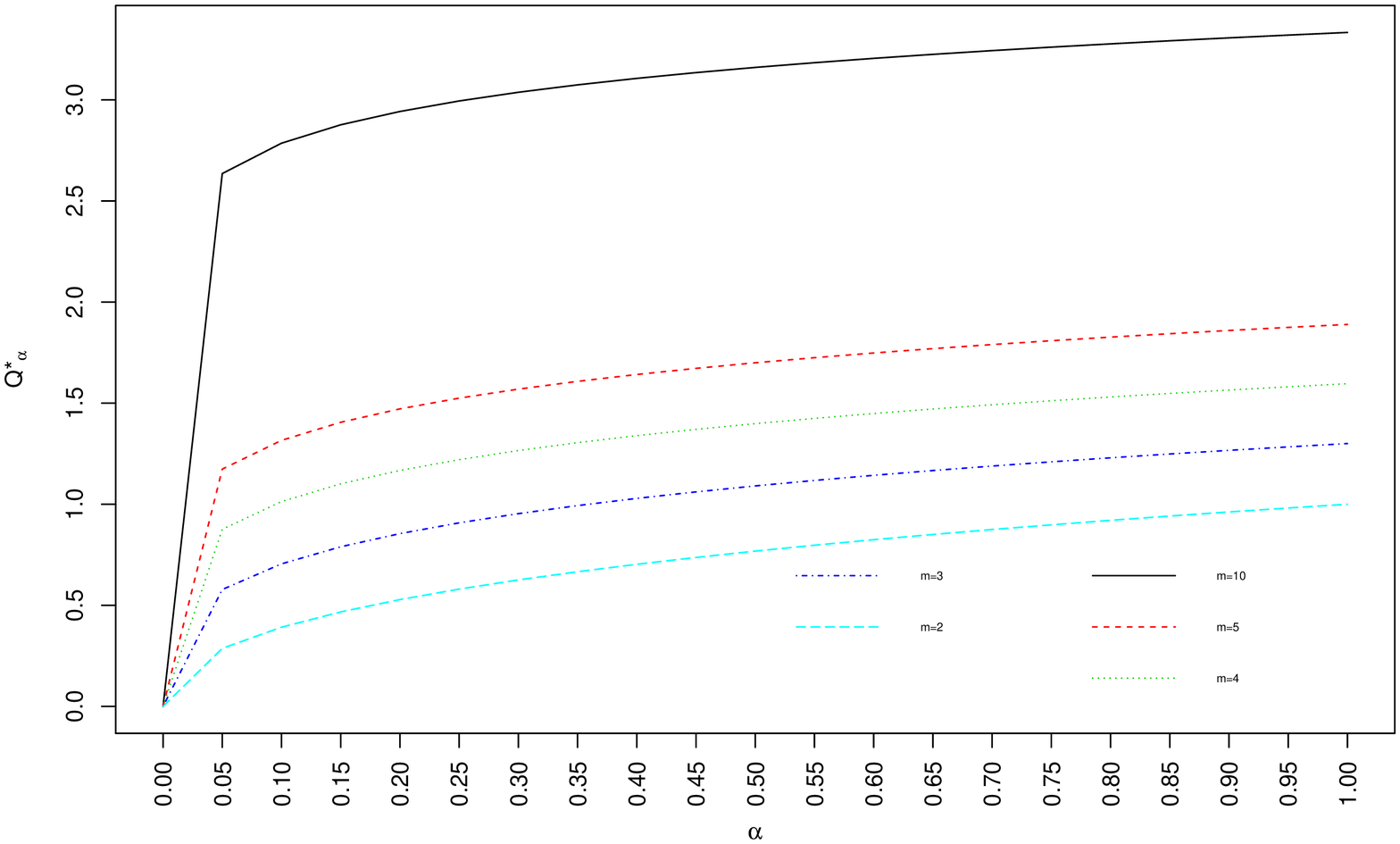}
	\caption{Optimal Order Quantity ($Q^*$) for different $\alpha_1\in(0,1)$ and $m$.  }
	\label{qstar}
\end{figure}
\end{document}